\documentclass[letterpaper, 10 pt, conference]{ieeeconf}
\IEEEoverridecommandlockouts
\usepackage{cite}
\usepackage{amsmath,amssymb,amsfonts}
\usepackage{algorithm} 
\usepackage{algpseudocodex}
\usepackage{graphicx}
\usepackage{textcomp}
\usepackage{xcolor}
\usepackage{url}
\usepackage{dblfloatfix}

\let\labelindent\relax

\def\BibTeX{{\rm B\kern-.05em{\sc i\kern-.025em b}\kern-.08em
    T\kern-.1667em\lower.7ex\hbox{E}\kern-.125emX}}
\usepackage{mathtools}

\usepackage{pgfplots}
\pgfplotsset{compat=newest}

\usepackage{enumerate}
\usepackage{enumitem}

\usepackage{acronym}
\acrodef{OLS}{ordinary least squares}
\acrodef{MLE}{maximum likelihood estimator}
\acrodef{SNR}{signal-to-noise ratio}
\acrodef{LTI}{linear time-invariant}
\acrodef{ARX}{autoregressive systems with exogenous inputs}
\acrodef{PE}{persistency of excitation}
\acrodef{i.i.d.}{independent, identically distributed}
\acrodef{QMI}{quadratic matrix inequality}
\acrodef{SME}{set-membership estimation}

\newtheorem{theorem}{Theorem}

\newtheorem{lemma}[theorem]{Lemma}
\newtheorem{remark}{Remark}
\newtheorem{assumption}{Assumption}
\newtheorem{prop}{Proposition}

\newcommand{\R}{\mathbb{R}}

\newcommand{\Prob}{\mathbb{P}}

\newcommand{\N}{\mathcal{N}}

\newcommand{\bigO}{\mathcal{O}}
\newcommand{\simiid}{\stackrel{\text{i.i.d.}}{\sim}}

\newcommand{\rev}[1]{\textcolor{black}{#1}}

\begin{document}
\title{Beyond Bounded Noise: Stochastic Set-Membership Estimation for Nonlinear Systems

\thanks{Funded by Deutsche Forschungsgemeinschaft (DFG, German Research Foundation) under Germany's Excellence Strategy - EXC 2075 – 390740016 and within grant AL 316/15-1 – 468094890. 
NC acknowledges the support of the Stuttgart Center for Simulation Science (SimTech).
FB acknowledges the support of the International Max Planck
Research School for Intelligent Systems (IMPRS-IS).}
\thanks{The authors are with the University of Stuttgart, Institute for Systems Theory and Automatic Control, 70550 Stuttgart,
Germany (e-mail:\{\mbox{felix.braendle}, \mbox{nicolas.chatzikiriakos}, \mbox{andrea.iannelli}, \mbox{frank.allgower}\}@ist.uni-stuttgart.de.}
}

\author{Felix Brändle, Nicolas Chatzikiriakos, Andrea Iannelli and Frank Allgöwer}

\maketitle

\begin{abstract}
In this paper, we derive a novel procedure for set-membership estimation of dynamical systems affected by stochastic noise with unbounded support.
Employing a bound on the sample covariance matrix, we are able to provide a finite-sample uncertainty set containing the true system parameters with high probability. 
Our approach can be natively applied to a wide class of nonlinear systems affected by sub-Gaussian noise.
Our analysis provides conditions under which the proposed uncertainty set converges to the true system parameters and establishes an upper bound on the convergence rate.
The proposed uncertainty set can be\;used directly for robust controller synthesis with probabilistic stability and performance guarantees. 
Concluding numerical examples demonstrate the advantages of the proposed formulation over established approaches.
\end{abstract}


\section{Introduction}
Many controller synthesis procedures require an accurate plant model. While such models can be derived from first principles, doing so often requires substantial expert knowledge.
On the other hand, system identification methods~\cite{ljung1999system} can be used to derive the model from data. 
For systems subject to stochastic noise with potentially unbounded support, early works employed correlation-based methods~\cite{care2017finite} to derive non-asymptotic uncertainty sets. More recently, advances in high-dimensional statistics~\cite{vershynin2012} have enabled the development of non-asymptotic guarantees for system identification~\cite{simchowitz2018learning, sarkar2019near,foster20a,Tsiamis2023}.
While these approaches do not naturally connect to robust control schemes, they have been leveraged for data-driven control, enabling end-to-end guarantees \cite{dean2020sample, chatzikiriakos2026a}. 
When bounds on the support of the noise affecting the system are available, \ac{SME} is often used since it provides an uncertainty description that naturally connects to robust control schemes. 
In essence, \ac{SME} directly yields an uncertainty set by considering all parameters consistent with the collected data~\cite{MILANESE2004957}.
In particular, the data-informativity framework as a subfield of \ac{SME} is attractive, as it focuses particularly on controlling and analyzing dynamical systems~\cite{waarde2022}. 
This framework provides a tight uncertainty formulation in the form of a \ac{QMI}~\cite{Waarde2023}.
\acp{QMI} can directly be used within robust control schemes to synthesize a guaranteed stabilizing controller for all systems in the uncertainty characterization \cite{dePersis2020, Waarde2023, berberich2020} and can easily be combined with other uncertainty descriptions \cite{braendle2025b}.
A question of particular interest is whether the uncertainty characterization obtained by \ac{SME} decreases as the number of samples increases and eventually converges to the true parameters.
This is particularly important for control applications since less uncertainty enables better controller performance. 
When the support of the noise is bounded, and the noise is i.i.d., \ac{SME} converges when the noise bound is tight~\cite{Yingying2024}. 
Notably, the convergence rate of \ac{SME} in this setting is faster than the fastest achievable convergence rate in settings with Gaussian noise \cite{Tsiamis2023}. 
The convergence rate of \ac{SME} for systems affected by stochastic noise with unbounded support is an open question.
Previous works tackling this setting either use high-probability bounds at each sampling instance~\cite{martin2023a} or a bound on the sample covariance of the noise~\cite{waarde2022}. 
However, the works above do not consider whether these methods converge to the true parameters and, if so, at what rate. 
\paragraph*{Contribution}
In this work, we provide a novel \ac{SME} procedure for unknown dynamical systems, \rev{which are linear in the unknown parameters and are affected by sub-Gaussian noise}. 
We leverage a bound on the sample covariance of the process noise to parameterize all possible parameter matrices consistent with the data.
Under standard \ac{PE} assumptions, we show asymptotic convergence of the identified parameter set to the true system matrices and provide a guaranteed bound on the convergence speed. 
To the best of our knowledge, this is the first work that provides these guarantees for \ac{SME} without requiring a bound on the support of the noise, but allowing for stochastic, unbounded noise.
Furthermore, our technique does not require evaluating the posterior distribution of a parameter matrix given the collected data, and therefore it can be easily applied to linear and nonlinear systems affected by sub-gaussian noise, \rev{as long as they are linear in the parameters}. 
Since we build on \ac{SME}, our approach can easily incorporate additional knowledge, e.g., on the system parameters, in the form of \acp{QMI}. 
Furthermore, \rev{as demonstrated in our numerical example,} the \ac{QMI} formulation can readily be used for robust controller design in a data-driven control scheme with probabilistic guarantees.
\paragraph*{Notation}
We denote matrix blocks that can be inferred from symmetry by
$\star$, i.e., we write $\Lambda^\top \Sigma \Lambda = [\star]^\top \Sigma \Lambda$. 
\rev{If $A\in\mathbb{R}^{n \times N}$ has full row rank, we use $A^\perp\in\mathbb{R}^{(N-n)\times N}$ to denote a matrix containing an orthonormal basis of the kernel of $A$ as rows, such that $A A^{\perp\top} = 0$ and $A^{\perp}A^{\perp\top}=I$}. 
Given a matrix $Z$ \rev{that has full row rank}, $Z^\dagger\rev{=Z^\top (ZZ^\top)^{-1}}$ denotes its Moore-Penrose inverse.
We write $f(N)=\mathcal{O}(g(N))$ if there exists
a constant $C>0$ such that $f(N)\le C g(N)$ for sufficiently large $N$.
For a matrix $A\in \R^{n\times n}$ we denote $A\preceq\mathcal{O}(g(N))$ if there exists a constant $C>0$ such that $A\preceq Cg(N)I_n$ where $I_n$ is the identity matrix of size $n$.
\section{Problem Setting}\label{sec:Setting}
We consider the unknown system 
\begin{equation}\label{eq:Sys}
    x_{t+1} = \theta_* z_t + w_t,
\end{equation}
where $x_t\in \R^{n_x}$ and $z_t = \phi(x_{t}, u_{t}) \in \R^{n_z}$ denote the state and lifted state at time $t$.  \rev{The lifting function $\phi$ is assumed to be known, and is allowed to be nonlinear.}
This leads to a linear regression model which includes LTI systems, but also systems with nonlinear liftings~\cite{mania2022active}.
The parameter matrix $\theta_*\in\mathbb{R}^{n_x \times n_z}$ is unknown and must be identified.
The input and stochastic process noise at time $t$ are denoted by $u_t\in \R^{n_u}$ and $w_t \in \R^{n_x}$. 
While the process noise is unknown, we impose the following assumption.
\begin{assumption}\label{ass:SubGausIsoNoise}
The process noise satisfies the following properties for all $t \ge 0$:
\begin{enumerate}[label=\roman*)]
    \item \rev{$w_t$ is drawn i.i.d. from a zero-mean  sub-gaussian distribution}, \label{add:SubGausIsoNoise:Subgaussian}
    \item The noise is isotropic, i.e., $\mathbb{E}[w_tw_t^\top] \rev{\coloneqq \Sigma_w} = I_{n_x}$. \label{add:SubGausIsoNoise:Isotopic}
\end{enumerate}
\end{assumption}
Sub-gaussian random variables can be used to model a broad class of stochastic noise, such as noise with bounded support and Gaussian noise.
Note that if the covariance matrix satisfies 
$\Sigma_w \succ 0$, 
Assumption\;\ref{ass:SubGausIsoNoise}\,\ref{add:SubGausIsoNoise:Isotopic} can be recovered by multiplying \eqref{eq:Sys} with $\Sigma_w^{-\text{\textonehalf}}$ from the left to satisfy Assumption\;\ref{ass:SubGausIsoNoise}.
\rev{To identify} the parameter matrix $\theta_*$ \rev{we have access to the} data $\mathcal{D}_N = \{\{z_t\}_{t=0}^N, \{x_t\}_{t=1}^{N+1}\}$.
Next, we introduce the following matrices
\begin{subequations}
\begin{align}
    X_N &= \begin{bmatrix}
        x_1 & \dots & x_{N}
    \end{bmatrix} \in \R^{n_x \times N}, \\
    Z_N &= \begin{bmatrix}
        z_0 & \dots & z_{N-1} 
    \end{bmatrix} \in \R^{n_z \times N}, \\
    W_N &=  \begin{bmatrix}
        w_0 & \dots & w_{N-1} 
    \end{bmatrix} \in \R^{n_x \times N},
\end{align}
\end{subequations}
where $X_N$ and $Z_N$ are from $\mathcal{D}_N$ and $W_N$ is the true, but unknown noise realization, which has generated the data.
This allows us to state \eqref{eq:Sys} in matrix notation
\begin{equation}
    X_N = \theta_* Z_N + W_N, \label{eq:Setup:MatrixNotation}
\end{equation}
which is more convenient for the upcoming analysis.
Since the data $\mathcal{D}_N$ only consists of finite samples affected by stochastic noise, it is generally not possible to recover $\theta_*$ exactly.
Hence, we aim to derive a set of parameter matrices, which must contain $\theta_*$ with a specified probability.
We employ a \ac{QMI} characterization
\begin{align} \label{eq:Setup:GeneralQMI}
    \begin{bmatrix}
        \theta^\top \\ I_{n_x} 
    \end{bmatrix}^\top  \Phi \begin{bmatrix}
        \theta^\top \\ I_{n_x} 
    \end{bmatrix} \succeq 0,
\end{align}
as it allows the direct application of robust control methods.
\rev{
In this work, we derive a \rev{matrix} $\Phi$ such that $\theta_*$ satisfies \eqref{eq:Setup:GeneralQMI} with high probability.
}
Taking the set defined by the \ac{QMI}~\eqref{eq:Setup:GeneralQMI} as uncertainty, it is possible to synthesize a controller, which is stabilizing for all possible parameter matrices $\theta$ in the set\cite{waarde2022,braendle2025a, chatzikiriakos2026a}.
Hence, the controller also stabilizes the true system with high probability, allowing us to provide probabilistic guarantees for the unknown system.
Furthermore, this framework allows us to analyze system properties such as passivity and to include performance objectives in the controller design, such as designing robust LQR-controllers \cite{Weiland1994}.
\rev{We demonstrate the applicability to controller design as part of our numerical evaluation in Section\,\ref{sec:numerics}.}

\section{Set-Membership estimation}
In this section, we derive a characterization of all parameter matrices consistent with a high probability set for the noise sample covariance.
Furthermore, we show under which conditions the set converges to the true parameters and provide a bound on the convergence rate.
\subsection{High-probability set-membership estimation}
To obtain a high-probability uncertainty set for the stochastic noise, recall the following result from high-dimensional statistics.
\begin{theorem}[Theorem 39, \cite{vershynin2012}] \label{theo:conv:WBound}
    Let $W^\top\in\mathbb{R}^{N\times n_x}$ be a matrix whose rows $w_i^\top$ are independent sub-gaussian isotropic random vectors with zero mean in $\mathbb{R}^{n_x}$. Then for every $N\geq 1$, with probability at least $1-\delta$ it holds that
    \begin{subequations}
    \begin{align}
        \sigma_{\max}(W^\top) \leq \sqrt{N} + c_1\sqrt{n_x} + \sqrt{\frac{1}{c_2} \log\left(\frac{2}{\delta}\right)},   \label{eq:conv:WSigmaMaxBound}\\
        \sigma_{\min}(W^\top) \geq \sqrt{N} - c_1\sqrt{n_x} - \sqrt{\frac{1}{c_2} \log\left(\frac{2}{\delta}\right)},   \label{eq:conv:WSigmaMinBound}
    \end{align}
    \end{subequations}
    where $c_1>0$ and $c_2>0$ are constants that depend only on the distribution of $w_i$ and $W=[w_1,\ldots,w_N]$.
    Furthermore, with the same probability, it holds that 
    \begin{align}
        \left\|\frac{1}{N} WW^\top - I_{n_x}  \right\| \leq \max(\eta, \eta^2), \label{eq:Conv:QuadraticBoundNoise}
    \end{align}
    where $\eta=c_1\sqrt{\frac{n_x}{N}} + \sqrt{\frac{\frac{1}{c_2} \log\left(\frac{2}{\delta}\right)}{N}}$.
\end{theorem}
Theorem\,\ref{theo:conv:WBound} provides a probabilistic bound on the sample covariance matrix of $w_t$ by bounding $\tfrac{1}{N}W_N W_N^\top$.
This is different from classical \ac{SME}, which employ a bound on the support of $w_t$ \cite{MILANESE2004957}.
The advantage of this approach is that it also holds for stochastic noise with unbounded support.
The following proposition provides a suitable high-probability set for the stochastic process noise and a corresponding high-probability uncertainty set for the parameters.
\begin{prop} \label{prop:SetW}
    Fix a failure probability $\delta\!\in\!(0,1)$.
    Let $W_N$ be generated according to Assumption\;\ref{ass:SubGausIsoNoise} and define the set 
    \begin{equation*}
    \mathcal{W}_N^{\kappa_\delta}\!\coloneq\!\left\{ W\in\mathbb{R}^{n_x \times N}\mid [\star]^\top \Phi_{W}(N,\kappa_\delta)
    \begin{bmatrix}
        W^\top \\ I_{n_x} 
    \end{bmatrix} \succeq 0\right\}
\end{equation*} 
with\; 
$\Phi_{W}(N,\kappa_\delta)\coloneqq\begin{bmatrix}
        -\frac{1}{N} I_{\rev{N}}  & 0 \\ 0 & \epsilon(N,\kappa_\delta)^2 I_{n_x} 
    \end{bmatrix}$
and $\epsilon(N,\kappa_\delta) \coloneqq 1\!+\!\frac{\kappa_\delta}{\sqrt{N}}$. 
Let $\kappa_\delta \ge c_1\sqrt{n_x} + \sqrt{\frac{1}{c_2} \log(\frac{2}{\delta})}$, where $c_1, c_2$ are the constants from Theorem\;\ref{theo:conv:WBound}.
\rev{Then, for all $N$ with probability at least $1-\delta$ it holds that
\begin{equation}
    W_N\in\mathcal{W}_N^{\kappa_\delta} \label{eq:Conv:PWn}   
\end{equation}
}
and consequently $\Prob\left[\theta_*\in\Theta_N^{\kappa_\delta}\right]\geq 1-\delta $ for
\begin{equation}
    \Theta_N^{\kappa_\delta} \coloneqq \left\{ \theta\in\mathbb{R}^{n_x \times n_z} \mid \exists W\!\in\!\mathcal{W}_N^{\kappa_\delta}\!:\: X_N \!= \!\theta Z_N\!+\!W\right\}\!. \label{eq:DefThetaN_delta}
\end{equation}
\end{prop}
\begin{proof}
    The case where the condition on $\kappa_\delta$ holds with equality follows directly from \eqref{eq:conv:WSigmaMaxBound} in Theorem\;\ref{theo:conv:WBound}. 
    Observing that Proposition\;\ref{prop:SetW} still holds when we overestimate the noise covariance completes the proof. 
\end{proof}
\begin{remark}
The constants $c_1, c_2$ can be determined by following the proof of \cite[Theorem 39]{vershynin2012}.
For $w_t \simiid \N(0, I_{n_x} )$ it holds that $c_1=1$ and $c_2=\tfrac{1}{2}$~\cite[Corollary 35]{vershynin2012}.
Note that while the existence of such constants is well established, exact characterizations or tight bounds are often not considered.
However, empirical estimates of $c_1$, $c_2$ could be used whenever it is possible to generate samples of $w_t$.
\end{remark}
The set $\Theta_N^{\kappa_\delta}$ is key for the upcoming analysis, as it provides a high probability set for $\theta_*$.
Furthermore, this set is independent of the lifting $z_t = \phi(x_t,u_t)$, allowing us to cover both LTI and nonlinear systems in the same framework.
By inserting $W^\top=X^\top - Z^\top \theta^\top$ in $\mathcal{W}_N^{\kappa_\delta}$, as in~\cite{Waarde2023}, we can \rev{provide the exact data-based expression of} $\Theta_N^{\kappa_\delta}$  
\begin{equation}\label{eq:Theta_N_delta_QMI}
    \begin{aligned}
    \Theta_N^{\kappa_\delta} = \Big\{ \theta&\in\mathbb{R}^{n_x \times n_z} \mid \\
    &[ \star ]^\top \Phi_{W}(N,\kappa_\delta)
    \begin{bmatrix}
        -Z_N^\top & X_N^\top \\ 0 & I_{n_x}  
    \end{bmatrix}
    \begin{bmatrix}
        \theta^\top \\I_{n_x}  
    \end{bmatrix}
     \succeq 0\Big\},
\end{aligned}
\end{equation}
where the matrices $X_N$ and $Z_N$ are known from $\mathcal{D}_N$.
\rev{As desired, this yields a \ac{QMI} description strucurally identical to \eqref{eq:Setup:GeneralQMI}}.
\rev{The dimension of the \ac{QMI} itself is independent of $N$, but to determine $\Phi$, we must perform matrix multiplications of $X_N$ and $Z_N$, whose size depends on $N$.}

\subsection{Asymptotic convergence}\label{sec:convergence}
Next, we consider, whether $\Theta_N^{\kappa_\delta}$ converges to $\theta_*$ for $N\to\infty$.
As shown in \cite{braendle2025a}, in an adversarial setting we cannot guarantee convergence for $N\to\infty$.
For intuition, consider a disturbance of the form $W_N = \bar{\theta}Z_N$, such that $X_N = (\theta_* + \bar{\theta}) Z_N$. 
Using only $X_N$ and $Z_N$ it is not possible to distinguish between $\theta_*$ and $\bar{\theta}$.
This is in contrast to the case of stochastic noise with bounded support. 
If there is sufficient probability mass at the boundary and the noise is i.i.d., it has been shown that \ac{SME} recovers the true parameters asymptotically~\cite{Yingying2024}.
In this section, we investigate whether noise satisfying Assumption\;\ref{ass:SubGausIsoNoise}, which is i.i.d. but allows for noise with unbounded support, allows us to recover $\theta_*$ nonetheless.
To do so, we first note an important relationship between $\Theta_N^{\kappa_\delta}$ and the \ac{OLS} estimate $\hat \theta_N^\mathrm{OLS} = X_N Z_N^\dagger$.
As shown in Lemma\,\ref{theo:conv:ThetaN} (Appendix) and~\cite{braendle2025a}, the center of the \ac{QMI} in $\Theta_N^{\kappa_\delta}$ corresponds to the \ac{OLS} estimate $\hat \theta^\mathrm{OLS}_N$.
Hence, for the convergence analysis, we require the following standard \ac{PE} assumption \cite{Tsiamis2023, Yingying2024}, for the \ac{OLS} estimate $\hat \theta^\mathrm{OLS}_N$ to converge to the true system parameters, i.e.,  $\lim_{N\to \infty} \hat \theta^\mathrm{OLS}_N = \theta_*$. 
\begin{assumption}\label{ass:PE_quad}
Fix a failure probability $\delta\in (0,1)$. There \rev{exist constants} $c_3(\delta)>0$, $c_4(\delta)>0$ such that for all $N$
\rev{with probability at least $1-\delta$
\vspace{-10pt}
\begin{equation}\label{eq:Conv:PECondition}
    c_3(\delta) I_{n_z}\preceq \frac{1}{N}\sum_{t=0}^{N-1} z_t z_t^\top \preceq c_4(\delta) I_{n_z}. 
\end{equation}}
\vspace{-10pt}
\end{assumption}
Assumption\;\ref{ass:PE_quad} ensures that $Z_N$ has full row rank.
For \ac{LTI} systems, the lower bound in Assumption\;\ref{ass:PE_quad} can easily be satisfied by, e.g., selecting isotropic Gaussian excitations during data collection \cite{simchowitz2018learning} or by using experiment design algorithms \cite{wagenmaker2020active, chatzikiriakos2025convex}. 
For nonlinear systems, recent works provide algorithms that establish PE data with high probability~\cite{mania2022active}.
The upper bound in Assumption\;\ref{ass:PE_quad} is a common condition in the system identification literature \cite{simchowitz2018learning, Tsiamis2023}, establishes statistical consistency of the \ac{OLS}, and holds naturally in many settings (cf. Remark\;\ref{rem:ConvergenceOLS} and \cite{sarkar2019near}).
\rev{
For the convergence analysis, we require that \eqref{eq:Conv:PWn} and \eqref{eq:Conv:PECondition} hold. 
However, each condition only holds with probability at least $1-\delta$ for each $N$, such that $\Theta_N^{\kappa_\delta}$ may be empty.
To account for this, we define a strictly increasing random sequence $N_1:\mathbb{N}\to\mathbb{N}$ of all $N$ that satisfy\;\eqref{eq:Conv:PWn} and\;\eqref{eq:Conv:PECondition}.
Using a union bound argument and choosing $\delta < \tfrac{1}{2}$, with probability at least $1\!-\!2\delta$ for each $N\in\mathbb{N}$ there exists a $n\in\mathbb{N}$ with $N_1(n)\!=\!N$. Thus, $N_1(n)$ is an infinite sequence almost surely, and we can consider $n\to\infty$ instead of $N\to\infty$.
Similarly, we define a strictly increasing sequence $N_2\!:\!\mathbb{N}\!\to\!\mathbb{N}$ of all $N$, which satisfy \eqref{eq:Conv:PECondition}.}
\begin{theorem}\label{theo:convergence}
    Suppose the \ac{OLS} estimate $\hat \theta^\mathrm{OLS}_N$ converges to $\theta_*$
    \rev{
    and let Assumptions~\ref{ass:SubGausIsoNoise}\,\ref{add:SubGausIsoNoise:Subgaussian} and~\ref{ass:PE_quad} be satisfied. 
    Further, let 
    $\Sigma_w\coloneq\lim_{N\to\infty}\tfrac{1}{N}W_{N}W_{N}^\top$ 
    and 
    fix a failure probability $\delta<\tfrac{1}{2}$.
    If $\kappa_{\delta}$ is selected 
    according to Proposition\;\ref{prop:SetW} assuming $\Sigma_w = I_{n_x}$, then the following hold almost surely:
    \begin{enumerate}[label=\alph*)]
        \item If $\Sigma_w = I_{n_x} $, then $\lim_{n\to\infty} \Theta_{N_1(n)}^{\kappa_\delta}=\theta_*$. \label{item:Conv:TheoExactConvergence}\vspace{1pt}
        \item If $\Sigma_w \preceq I_{n_x}$, then $\theta_* \in \Theta_{N_1(n)}^{\kappa_\delta}$ for $n\to\infty$.
        \label{item:Conv:TheoOverEstimate}
        \item If $\Sigma_w \npreceq I_{n_x} $, then $\lim_{n\to\infty}\Theta_{N_2(n)}^{\kappa_\delta}=\emptyset$. \label{item:Conv:TheoUnderEstimate}
    \end{enumerate}
    }
\end{theorem}
\begin{proof}
First note that, by Lemma\;\ref{theo:conv:ThetaN} (Appendix) we can express $\Theta_{N}^{\kappa_\delta}$ by the following inequality
\vspace{-2pt}
    \begin{align*}
        \frac{1}{N}(\theta\!-\! {\hat{\theta}^\mathrm{OLS}_{N}})Z_NZ_N^\top &(\theta\!-\! {\hat{\theta}^\mathrm{OLS}_{N}})^\top \! \\ &\preceq \epsilon(N, \kappa_\delta)^2I_{n_x}  \!-\! \tfrac{1}{N} \theta_0 Z_N^\perp Z_N^{\perp \top} \theta_0^\top  
    \end{align*}
    with $\theta_0=X_N Z_N^{\perp \top}$\rev{ and $\hat{\theta}^\mathrm{OLS}_{N}=X_NZ_N^\dagger$.
    Since $\begin{bmatrix}  Z_N^\dagger & Z_N^{\perp\top} \end{bmatrix}$ is the inverse of $\begin{bmatrix} Z_N \\ Z_N^\perp        \end{bmatrix}$, we obtain
    \vspace{-5pt}
    \begin{align}\label{eq:conv:OLS_System}
        X_N=X_N\begin{bmatrix}
            Z_N^\dagger & Z_N^{\perp\top}
        \end{bmatrix}\begin{bmatrix}
            Z_N \\ Z_N^\perp
        \end{bmatrix}=\hat{\theta}^\mathrm{OLS}_{N} Z_N + \theta_0 Z^{\perp }_N.
    \end{align}
    Next, we define the difference between the ordinary least squares estimate and the true parameter matrix as $\Delta_\theta \coloneqq \theta_* - \hat{\theta}^\mathrm{OLS}_{N}$.
    By comparing \eqref{eq:Setup:MatrixNotation} and \eqref{eq:conv:OLS_System}, we arrive at $\theta_0 Z^{\perp}_N = \Delta_\theta Z_N + W_N$.
    Thus, any $\theta \in \Theta_N^{\kappa_\delta}$ satisfies
    \begin{align*}
        \tfrac{1}{N}(\theta\!&-\! {\hat{\theta}^\mathrm{OLS}_{N}})Z_NZ_N^\top (\theta\!-\! {\hat{\theta}^\mathrm{OLS}_{N}})^\top \! \\ &\preceq \epsilon(N, \kappa_\delta)^2I_{n_x}  \!-\! \tfrac{1}{N} (\Delta_\theta Z_N + W_N) (\Delta_\theta Z_N + W_N)^\top.
    \end{align*}
    We now consider the sequences $N_i(n)$, $i =1,2$. It holds
    \begin{equation}
    \begin{aligned} \label{eq:conv:TheoProof:LimitInequality}
    \lim_{n\to \infty} &\tfrac{1}{N_i(n)}(\theta\!-\! {\hat{\theta}^\mathrm{OLS}_{N_i(n)}})Z_{N_i(n)}Z_{N_i(n)}^\top (\theta\!-\! {\hat{\theta}^\mathrm{OLS}_{N_i(n)}})^\top \! \\ &- \epsilon(N_i(n), \kappa_\delta)^2I_{n_x} \! +\! \tfrac{1}{N_i(n)} W_{N_i(n)}W_{N_i(n)}^\top \preceq  0,
    \end{aligned}   
    \end{equation}
    where we used that the \ac{OLS} estimate converges, i.e., $\lim_{n\to \infty} \Delta_\theta = 0$ while along the sequences $N_i(n)$ the statement of Assumption~\ref{ass:PE_quad} holds. 
    Now, we consider
    \begin{equation} \label{eq:conv:TheoProof:LimitConvergingPart}
        \lim_{n\to\infty}\tfrac{1}{N_i(n)} W_{N_i(n)}W_{N_i(n)}^\top -\epsilon(N_i(n), \kappa_\delta)^2I_{n_x}.
    \end{equation}
    Result \ref{item:Conv:TheoExactConvergence} follows by leveraging $\lim_{n\to\infty}\epsilon(N_i(n), \kappa_\delta)^2=1$, and \eqref{eq:conv:TheoProof:LimitConvergingPart} being negative semi-definite for all $N_1(n)$, due to \eqref{eq:Conv:PWn}.
    From \eqref{eq:Conv:PECondition} follows that $\tfrac{1}{N_i(n)}Z_{N_i(n)}Z_{N_i(n)}^\top$ is positive definite, hence only $\theta_*$ can satisfy \eqref{eq:conv:TheoProof:LimitInequality}.
    Result \ref{item:Conv:TheoOverEstimate} follows the same steps, where $\eqref{eq:conv:TheoProof:LimitConvergingPart}$ does not approach zero for $\Sigma_w\neq I_{n_x}$. 
    Result \ref{item:Conv:TheoUnderEstimate} follows by realizing that \eqref{eq:conv:TheoProof:LimitConvergingPart} is not negative semi-definite, such that \eqref{eq:conv:TheoProof:LimitInequality} can not be satisfied by any $\theta$.
    }
\end{proof}
Theorem\;\ref{theo:convergence}\;\ref{item:Conv:TheoOverEstimate} and \ref{item:Conv:TheoUnderEstimate}  provide results for the cases where the true noise covariance is over- \rev{($\mathbb{E}[w_tw_t^\top] \preceq I_{n_x}$)} or underestimated \rev{($\mathbb{E}[w_tw_t^\top] \npreceq I_{n_x}$)}.  
If the variance of $w_t$ is overestimated, $\Theta^{\kappa_\delta}_N$ may not converge to a unique parameter matrix.
In contrast, if the true variance is larger than assumed, $\Theta^{\kappa_\delta}_N$ will lead to the empty set. 
\rev{In practical applications,} this provides additional insights into the validity of the assumed noise covariance. 
\begin{remark}\label{rem:ConvergenceOLS}
    Note that Theorem\;\ref{theo:convergence} requires convergence of the \ac{OLS} estimate. For many collection schemes and system classes, Assumption\;\ref{ass:PE_quad} is sufficient to ensure convergence of the \ac{OLS} estimate~\cite{Tsiamis2023, dean2020sample, chatzikiriakos2026a}. In fact, as shown by \cite{sarkar2019near}, when the data is collected from a trajectory of an unstable linear system, the \ac{OLS} is statistically inconsistent. In this case, the data does not satisfy Assumption~\ref{ass:PE_quad}.  
    While this inconsistency is not observed in the case of a finite hypothesis class \cite{chatzikiriakos2024b}, addressing this inconsistency for the estimation of unstable LTI systems from a single trajectory is an open problem. 
\end{remark}

\subsection{Convergence speed} \label{sec:convRate}
In the previous section, we showed convergence of $\Theta_N^{\kappa_\delta}$ for $\lim_{N\to\infty}\frac{1}{N}W_NW_N^\top=I_{n_x}$. 
Now, we shift our attention to the convergence rate.
To provide an upper bound, we require the following assumption. 
\begin{assumption}\label{ass:regularity}
    For a failure probability $\delta \in (0,1)$. There exists a constant
    $c_5(\delta)>0$, such that with probability at least $1-\delta$ the \ac{OLS} estimate $\hat\theta_N^\mathrm{OLS}$ satisfies
    \begin{equation}
        \|\hat{\theta}_N^\mathrm{OLS} - \theta_*\|^2 \leq \frac{c_5(\delta)}{N}.
    \end{equation} 
\end{assumption}
Note that Assumption\;\ref{ass:regularity} is the convergence rate obtained by \ac{OLS} for many different system classes and data collection schemes \cite{simchowitz2018learning, dean2020sample,Tsiamis2023, chatzikiriakos2026a, foster20a}.
Given Assumption\;\ref{ass:regularity}, we provide an upper bound on the convergence rate.
\begin{theorem}
    \label{prop:COnvergenceRate}
    Suppose Assumptions\,\ref{ass:SubGausIsoNoise}--\ref{ass:regularity} hold and fix $\delta\in(0,1)$. If $\kappa_\delta$ is chosen as in Proposition~\ref{prop:SetW}, then with probability at least $1 - 3\delta$ it holds that
    \begin{align} \label{eq:statementConvRate}
        \Theta^{\kappa_\delta}_N \subseteq \left\{\theta\in\mathbb{R}^{n_x \times n_z}\mid \|\theta - \hat{\theta}_N^\mathrm{OLS}\|^2 \leq \bigO \left(\tfrac{1}{\sqrt{N}}\right)\right\}.
    \end{align}
\end{theorem}\vspace{2pt}
\begin{proof}
    For this proof, we define the events 
    \begin{subequations}
        \begin{alignat*}{2}
            \mathcal{E}_1 &\coloneqq \!\left\{  1-\!\tfrac{\kappa_\delta}{\sqrt{N}} \!\le \!\tfrac{\sigma_\mathrm{min}(W^\top)}{\sqrt{N}} \!\le\! \tfrac{\sigma_\mathrm{max}(W^\top)}{\sqrt{N}} \!\le 1+\!\tfrac{\kappa_\delta}{\sqrt{N}} \right\}, \\
            \mathcal{E}_2 &\coloneqq \bigg\{c_3(\delta)I_{n_z} \preceq \tfrac{1}{N}\sum_{t=0}^{N-1} z_t z_t^\top \preceq c_4(\delta) I_{n_z}\bigg\}  , \\
            \mathcal{E}_3 &\coloneqq \bigg\{\|\hat{\theta}_N^\mathrm{OLS} - \theta_*\| \leq \tfrac{\sqrt{c_5(\delta)}}{\sqrt{N}}\bigg\}.
        \end{alignat*}
    \end{subequations} 
    Note that $\Prob[\mathcal{E}_1] \ge 1-\delta$, by Theorem\;\ref{theo:conv:WBound} since Assumption\;\ref{ass:SubGausIsoNoise} holds. 
    Further, since $\Prob[\mathcal{E}_2] \ge 1-\delta$ and $\Prob[\mathcal{E}_3] \ge 1-\delta$, by Assumptions\;\ref{ass:PE_quad} and~\ref{ass:regularity}, respectively we can use union bound arguments to obtain that $\Prob[\bigcap_{i =1}^3\mathcal{E}_i] \ge 1-3\delta$.
    For the remainder of this proof, we will show that when the events $\mathcal{E}_1$, $\mathcal{E}_2$ and $\mathcal{E}_3$ hold simultaneously, then \eqref{eq:statementConvRate} holds.
    Recall that from the proof of Theorem\;\ref{theo:convergence} we have 
     \begin{align*}
        &\frac{1}{N}(\theta\!-\! {\hat{\theta}^\mathrm{OLS}_{N}})Z_NZ_N^\top (\theta\!-\! {\hat{\theta}^\mathrm{OLS}_{N}})^\top \! \\ &\preceq \epsilon(N, \kappa_\delta)^2I_{n_x}  \!-\! \tfrac{1}{N} (\Delta_\theta Z_N + W_N) (\Delta_\theta Z_N + W_N)^\top \\
        &\preceq \epsilon(N, \kappa_\delta)^2I_{n_x} \!- \!\tfrac{1}{N}( W_N W_N^\top  \rev{+}  W_N Z_N^\top\Delta_\theta^\top \rev{+} \Delta_\theta Z_NW_N^\top),
    \end{align*}
    where in the last step we used the fact that $\Delta_\theta Z_N Z_N^\top\Delta_\theta^\top$ is positive semi-definite. 
    Now note that, under $\mathcal{E}_1$ we have 
    \begin{align*}
        \epsilon&(N, \kappa_\delta)^2I_{n_x}  -\tfrac{1}{N}W_NW_N^\top \\ &\preceq \left(\left( 1 + \tfrac{\kappa_\delta}{\sqrt{N}}\right)^2 - \left( 1 - \tfrac{\kappa_\delta}{\sqrt{N}}\right)^2 \right) I_{n_x} \preceq \bigO\left(\tfrac{1}{\sqrt{N}}\right).
    \end{align*}
    Further, under the event $\mathcal{E}_1$, the upper bound in $\mathcal{E}_2$, and $\mathcal{E}_3$ we can bound the individual terms to obtain 
    \begin{align*}
        \rev{-}\tfrac{1}{N}( W_N Z_N^\top\Delta_\theta^\top \rev{+} \Delta_\theta Z_NW_N^\top)\preceq \bigO\left(\tfrac{1}{\sqrt{N}}\right).
    \end{align*}
    Putting everything together, we have 
    \begin{align*}
        \tfrac{1}{N}(\theta\!-\! {\hat{\theta}^\mathrm{OLS}_{N}})Z_NZ_N^\top &(\theta\!-\! {\hat{\theta}^\mathrm{OLS}_{N}})^\top \preceq \bigO\left(\tfrac{1}{\sqrt{N}}\right).
    \end{align*}
    From the lower bound in $\mathcal{E}_2$ it follows
    \begin{align*}
        (\theta\!-\! {\hat{\theta}^\mathrm{OLS}_{N}})(\theta\!-\! {\hat{\theta}^\mathrm{OLS}_{N}})^{\!\top} \!\!\!\preceq\! \tfrac{1}{c_3(\delta)N}(\theta\!-\! {\hat{\theta}^\mathrm{OLS}_{N}})Z_NZ_N^\top(\theta\!-\! {\hat{\theta}^\mathrm{OLS}_{N}})^{\!\top} \!.
    \end{align*}  
    By combining both lower and upper bounds, we obtain that 
        
        \begin{equation}
            \theta\in\Theta_N^{\kappa_\delta} \implies (\theta\!-\! {\hat{\theta}^\mathrm{OLS}_{N}})(\theta\!-\! {\hat{\theta}^\mathrm{OLS}_{N}})^\top \preceq \bigO\left(\tfrac{1}{\sqrt{N}}\right).
        \end{equation}
\end{proof} 
While Theorem\;\ref{theo:convergence} showed convergence of the proposed \ac{SME} approach to the true parameter matrix,
Theorem\,\ref{prop:COnvergenceRate} provides a bound on the convergence rate by showing that the size of $\Theta_N^{\kappa_\delta}$ decreases with at least $\bigO(\frac{1}{\sqrt{N}})$. 
Note that this is notably slower than the \ac{OLS} rate itself for many known settings, which converges with $\bigO(\frac{1}{N})$, as discussed in Assumption\;\ref{ass:regularity}. 
This leads to a timescale separation in the different convergence rates.
First, the \ac{OLS}-estimate $\hat{\theta}^\mathrm{OLS}_N$ converges to $\theta_*$, and then $\Theta_N^{\kappa_\delta}$, which is centered around $\hat{\theta}^\mathrm{OLS}_N$ converges to a singleton. 
While the rate of convergence of our approach is slower than other approaches, which are directly applicable for controller design of linear systems, such as \cite{umenberger2019, Waarde2023}, our approach offers different advantages (cf. our numerical evaluations in Section\;\ref{sec:numerics}).
For one, our formulation is easily applicable to nonlinear systems, as we require only the bound on the noise matrix $W_N$, which is indepedent on the nonlinear lifting $\phi(x_t,u_t)$.
Furthermore, Theorem\;\ref{theo:convergence}\,\ref{item:Conv:TheoOverEstimate} and \ref{item:Conv:TheoUnderEstimate} provide a built-in hypothesis test. 
Over- or underestimating the sample covariance leads to a qualitative change in $\Theta_N^{\kappa_\delta}$, as it either becomes empty or does not converge to a singleton.
To the best of our knowledge, this is the first result that combines \ac{SME} with unbounded, stochastic noise, while demonstrating convergence and providing a bound on the convergence rate.
Further improvements may be possible by including the lower bound \eqref{eq:conv:WSigmaMaxBound} in the construction of $\mathcal{W}_N^{\kappa_\delta}$.
\section{Numerical evaluation} \label{sec:numerics}
\rev{To validate and contextualize our theoretical results, we provide two numerical experiments.}\footnote{The Matlab code for the numerical examples can
be accessed at: \texttt{https://github.com/col-tasas/2026-stochastic-SME}}
\rev{First, we consider a data-driven control setting with an LTI system.
Second, we investigate the convergence properties of the proposed \ac{SME} method for a nonlinear example.}

\emph{Controller synthesis:}
\rev{
Our framework readily integrates with data-driven control, as it is based on a \ac{QMI} parameterization that can be directly incorporated into existing robust controller design schemes. Moreover, the proposed method detects violations of the assumed noise covariance. In contrast, \cite{dePersis2020,berberich2020,Waarde2023,umenberger2019} can yield overly small uncertainty regions and consequently invalid end-to-end guarantees without detecting such violations. Finally, unlike \cite{Waarde2023, umenberger2019}, our method also applies to nonlinear systems that are linear in the parameters.
}
\rev{
To illustrate these findings, consider a second-order linear oscillator
\begin{align*}
    x_{t+1} = \begin{bmatrix}
    1&1\\-0.1&0.8
    \end{bmatrix} x_t + 
    \begin{bmatrix}
        0\\1
    \end{bmatrix}u_t + w_t, \quad  w_t \simiid \mathcal{N}(0,\sigma^2I_2 )
\end{align*}
with noise covariance $\sigma\in\{1;1.1\}$. 
We consider all parameters of the system unknown and consequently define $z_t = \begin{bmatrix}
    x_t^\top  & u_t
\end{bmatrix}^\top$. 
To generate data, we fix the initial condition $x_0=0$ and use a fixed input sequence sampled according to $u_t\simiid \mathcal{N}(0,\sigma_u^2)$ with $\sigma_u=5$ to generate a single trajectory.} 
\rev{
We compare our approach to two existing approaches.
}
First, consider the approach presented in \cite{dePersis2020,berberich2020}. 
This approach also uses \ac{SME} techniques for controller synthesis with $\mathcal{W}_N^{\kappa_\delta}$, but neglects the additional constraint that any $W_N\in\mathcal{W}_N^{\kappa_\delta}$ must also satisfy $X_N = \theta Z_N + W_N$ by using $W_NZ_N^\dagger$ as disturbance (see \cite{braendle2025a} for more details). 
The resulting set is defined as
\begin{align*}
    \Theta^\delta_{1}\!=\!\left\{\theta \mid \!\tfrac{1}{N}(\theta - \hat{\theta}_N^{\mathrm{OLS}})^\top (\theta - \hat{\theta}_N^{\mathrm{OLS}}) \!\preceq\!  \epsilon(N, \kappa_\delta)^2 Z_N^{\dagger\top}\!Z_N^\dagger\right\}
\end{align*}
\rev{where $\epsilon(N, \kappa_\delta)$ is defined in Proposition\;\ref{prop:SetW}.}
To get a \ac{QMI} in $\theta^\top$ instead of $\theta$, we apply the dualization lemma \cite[Lemma 4.9]{Weiland1994}.
Secondly, we consider the scheme proposed in \cite[Proposition 2.1]{umenberger2019} and used in \cite[Section 5.4]{Waarde2023}, which also provides a \ac{QMI} description according to
\begin{align*}
    \Theta^\delta_{2}=\left\{\theta \mid (\theta - \hat{\theta}_N^{\mathrm{OLS}}) Z_NZ_N^\top(\theta - \hat{\theta}_N^{\mathrm{OLS}})^\top \preceq c_\delta I_{n_x} \right\},
\end{align*}
where $c_\delta$ is the quantile function for the failure probability $\delta$ with respect to the Chi-squared distribution.
Note that this approach is only applicable to linear dynamics.
\rev{For our analysis, we fix the failure probability $\delta=0.05$ and perform a Monte Carlo simulation by constructing $\Theta_N^{\kappa_\delta}$ $100$ times for each $\sigma\in\{1;1.1\}$.
The sets $\Theta_N^{\kappa_\delta}$, $\Theta^\delta_{1}$, and $\Theta^\delta_{2}$ all assume isotropic noise ($\sigma\!=\!1$), so we can investigate the effect of the assumed noise covariance bound being wrong.
Given the uncertainty sets, we employ \cite[Theorem 17]{waarde2022} to synthesize a data-driven $\mathcal{H}_2$-statefeedback controller $u_t = K_N x_t$ using MOSEK \cite{mosek} for different $N$.
The controller minimizes the worst-case $\mathcal{H}_2$-norm from $w_t$ to the output $[x_t,u_t]^\top$.
For more details regarding robust controller synthesis, see \cite{Weiland1994}.
We compare it to the nominal $\mathcal{H}_2$-statefeedback controller of the true system matrices with controller gain $K_\mathrm{*}$. 
The difference $\|K_N-K_\mathrm{*} \|$ for different $N$ and $\sigma$ are visualized in Fig.\,\ref{fig:numerical:UncertaintyVolume}.}
\begin{figure}[t]
    \centering
    \includegraphics[width = 0.92\linewidth]{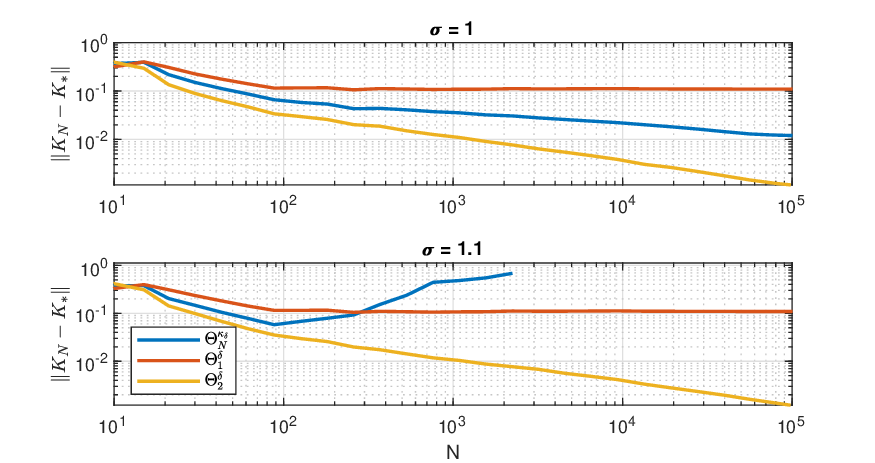}
     \vspace{-9pt}
    \caption{Difference between nominal and data-driven controllers.}
    \label{fig:numerical:UncertaintyVolume}
    \vspace{-18pt}
\end{figure}
\rev{For $\sigma=1$, we see that our approach and $\Theta^\delta_{2}$ converge to the true optimal controller matrix $K_\mathrm{*}$, while for $\Theta^\delta_{1}$ the controller gain $K$ does not converge to $K_\mathrm{*}$.
We notice that $\Theta^\delta_{2}$ converges at a faster rate than our approach, as it uses the bound of the \ac{OLS}-estimate for linear systems, similar to Assumption\;\ref{ass:regularity}. 
As expected, when $\sigma=1.1$, i.e., $\sigma$ is larger than the assumed covariance $\sigma = 1$, our method detects this violation of the noise assumption and yields the empty set. For the approaches $\Theta^\delta_{1}$ and $\Theta^\delta_{2}$, $K_N$ shows the same behavior as for the case $\sigma=1$.
This is problematic, since for finite $N$ the obtained uncertainty sets are no longer guaranteed to contain the true system with the specified probability, since the noise assumptions are violated.
In this example, when $\Theta^{\kappa_\delta}_N$ becomes empty, we empirically have $\Prob\left[\theta_*\in\Theta^\delta_{2}\right]=0.93$ instead of the specified $0.95$.
Hence, the robust controller synthesis methods will not provide the desired performance and stability guarantees.
Although our method also does not provide a controller with the specified stability guarantees, it is possible to detect this case as $\Theta_N^{\kappa_\delta}$ becomes empty, such that the assumed covariance can be adjusted accordingly.}

\rev{\emph{Nonlinear system:}
Next, we apply our parameterization to a nonlinear system
\begin{equation*}
    x_{t+1} = 0.9 x_t - 1.5\sin(x_t) + 0.5 u_t + w_t
\end{equation*}
with state $x_t$, noise $w_t \simiid \mathcal{N}(0,\sigma^2)$, nonlinear lifting $\phi(x_t,u_t)=[x_t,\sin(x_t),u_t]^\top$, unknown parameter vector $\theta_*=[0.9,\,-1.5,\,0.5]$ and $\sigma\in\{0.09;0.1;0.11\}$.
As before, we fix the initial condition $x_0=0$ and excite the system with $u_t\simiid \mathcal{N}(0,\sigma_u^2)$ with $\sigma_u=0.5$, fix the failure probability $\delta=0.05$.
and repeat the experiment $200$ times for each $\sigma$.
We treat $\sigma = 0.1$ as the known covariance and as described in Section\,\ref{sec:Setting}, we scale $X_N$ and $Z_N$ by a factor $0.1$ to recover the isotropic case for the nominal covariance $0.1^2$.
We compare our approach only to $\Theta^\delta_{1}$, since $\Theta^\delta_{2}$ can only be used for linear systems.
To this end, we compute the volume of the ellipses describing each set and evaluate the mean.
The results are shown in Fig.\,\ref{fig:numerical:UncertaintyVolumeNonlinear}.
\begin{figure}[t]
    \centering
    \includegraphics[width = 0.92\linewidth]{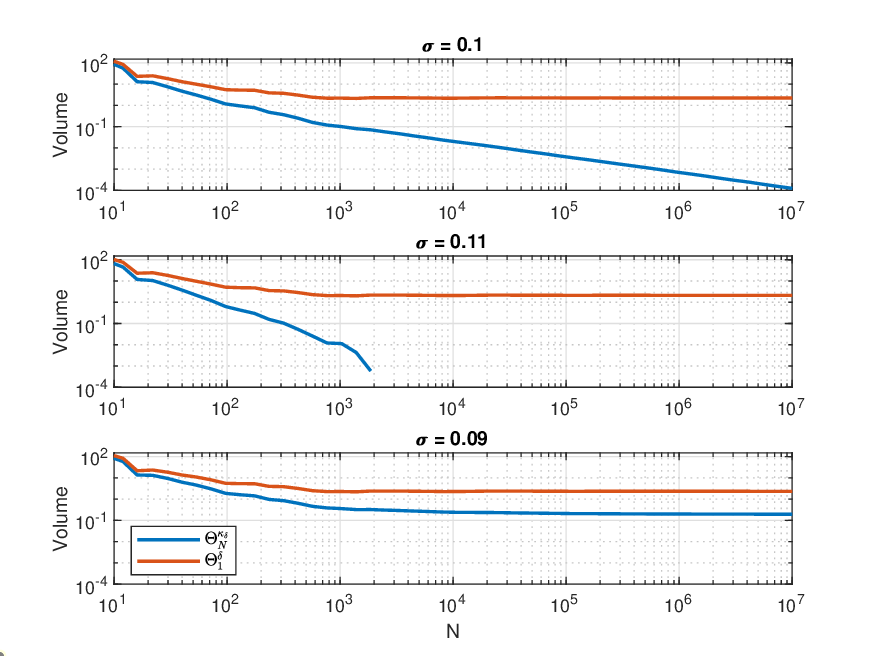}
     \vspace{-14pt}
    \caption{Mean volume of uncertainty sets  for the nonlinear system.}
    \label{fig:numerical:UncertaintyVolumeNonlinear}    
    \vspace{-14pt}
\end{figure}
Note that, even though we consider a nonlinear system, for $\sigma = 0.1$ and $\sigma = 0.11$ we observe the same trends as in the previous example.
For $\sigma = 0.09$, both $\Theta_N^{\kappa_\delta}$ and $\Theta^\delta_{1}$ can not recover $\theta_*$, (cf. Theorem~\ref{theo:convergence}~\ref{item:Conv:TheoOverEstimate}).
When investigating $N\in[10^{3},10^{7}]$ for $\sigma =1$, the volume decreases by about three decades, indicating a convergence rate of $\bigO(N^{-\frac34})$.
As the volume of a three-dimensional ball grows cubically with its radius, this is a numerical indication of the tightness of our upper bound on the convergence rate.}
\section{Conclusion}
In this work, we proposed a parameterization to include unbounded noise into the \ac{SME} framework. 
To do so, we employ a general bound on the sample covariance and show convergence under standard assumptions on persistency of excitation. 
Our approach can be applied to nonlinear systems and provides a qualitative way to validate whether the assumptions on the noise are justified. \rev{Future work includes using this mechanism to co-identify the unknown parameters and noise statistics.}
In addition, investigating whether the convergence rate can be improved, e.g., by including the lower bound~\eqref{eq:conv:WSigmaMinBound}, \rev{as well as leveraging targeted exploration to shape the set $\Theta_N^{\kappa_\delta}$ present promising research directions.}

\appendix 
\begin{lemma}\label{theo:conv:ThetaN}
    Let Assumption\;\ref{ass:PE_quad} be satisfied. Then the set $\Theta^{\kappa_\delta}_N$ defined in Proposition\,\ref{prop:SetW} can exactly be described by 
    \begin{align*}
        \Theta_N^{\kappa_\delta} = \left\{ \theta\in\mathbb{R}^{n_x \times n_z} \mid [ \star ]^\top \Phi_{\Theta}(N)\begin{bmatrix}
        (\theta- {\hat{\theta}^\mathrm{OLS}_{N}})^\top \\ I_{n_x} 
    \end{bmatrix}\succeq 0\right\}
    \end{align*}
    with 
    \begin{equation*}
        \Phi_{\Theta}(N) = \begin{bmatrix}
            -\frac{1}{N}Z_NZ_N^\top & 0 \\
            0 & \epsilon(N,\kappa_\delta)^2I_{n_x}  - \frac{1}{N} \theta_0 Z_N^\perp Z_N^{\perp \top} \theta_0^\top
        \end{bmatrix},
    \end{equation*}
    \rev{$\theta_0=X_NZ_N^{\perp \top}$} and $\hat\theta^\mathrm{OLS}_{N}=X_N Z_N^\dagger$ being the \ac{OLS} estimate.
\end{lemma}
\begin{proof}
    This follows similar steps as in \cite[Corollary 1]{braendle2025a}. First, we take any $\theta\in\Theta_N^{\kappa_\delta}$ and use the decomposition 
    \rev{$W_N = \theta_0 Z_N^\perp + \theta_1 Z_N$ to obtain $X_N = (\theta+\theta_1)Z_N + \theta_0Z_N^\perp.$}
    Now by multiplication with $[Z_N^\dagger\,Z_N^{\perp\top}]$ and due to $Z_NZ_N^\dagger=I$, $Z_N^\perp Z_N^\dagger=0$, this results in \rev{$\theta + \theta_1=X_N Z_N^\dagger=\hat\theta_N^\mathrm{OLS}$ and $\theta_0=X_NZ_N^{\perp \top}$}. 
    This leads to $W_N^\top = Z_N^\top (\hat \theta_N^\mathrm{OLS}- \theta)^\top + Z_N^{\perp\top} \theta_0^\top$, \rev{ which by plugging it into~\eqref{eq:DefThetaN_delta} yields} 
    \begin{equation*}
        [\star] \begin{bmatrix}
            -\frac1N I_{\rev{N}} & 0 \\ 0 & \epsilon(N, k_\delta) ^2 I_{n_x} 
        \end{bmatrix}
        \begin{bmatrix}
            Z_N^\top (\hat \theta_N^\mathrm{OLS}- \theta)^\top + Z_N^{\perp^\top} \theta_0^\top \\ I_{n_x}  
        \end{bmatrix}.
    \end{equation*}
    Elementary reformulations recover the result.
\end{proof}

\bibliographystyle{IEEEtran}
\bibliography{references.bib}

@InProceedings{simchowitz2018learning,
  author       = {Simchowitz, M. and Mania, H. and Tu, S. and Jordan, M. I. and Recht, B.},
  booktitle    = {Conf. On Learning Theory},
  title        = {Learning without mixing: Towards a sharp analysis of linear system identification},
  year         = {2018},
  pages        = {439--473},
}

@InProceedings{wagenmaker2020active,
  author       = {Wagenmaker, A. and Jamieson, K.},
  booktitle    = {Conf. on Learning Theory},
  title        = {Active learning for identification of linear dynamical systems},
  year         = {2020},
  organization = {PMLR},
  pages        = {3487--3582},
}

@InProceedings{sarkar2019near,
  author       = {Sarkar, T. and Rakhlin, A.},
  booktitle    = {International Conf. on Machine Learning},
  title        = {Near optimal finite time identification of arbitrary linear dynamical systems},
  year         = {2019},
  organization = {PMLR},
  pages        = {5610--5618},
}

@Article{dean2020sample,
  author    = {Dean, S. and Mania, H. and Matni, N. and Recht, B. and Tu, S.},
  journal   = {Foundations of Computational Mathematics},
  title     = {On the sample complexity of the linear quadratic regulator},
  year      = {2020},
  number    = {4},
  pages     = {633--679},
  volume    = {20},
  publisher = {Springer},
}

@Article{Tsiamis2023,
  author    = {Tsiamis, A. and Ziemann, I. and Matni, N. and Pappas, G. J.},
  journal   = {IEEE Control Systems},
  title     = {Statistical Learning Theory for Control: A Finite-Sample Perspective},
  year      = {2023},
  number    = {6},
  pages     = {67--97},
  volume    = {43},
}

@Book{ljung1999system,
  author    = {Ljung, L.},
  publisher = {Prentice Hall PTR},
  title     = {System Identification: Theory for the User},
  year      = {1999},
}

@Article{chatzikiriakos2024b,
      title={Sample Complexity Bounds for Linear System Identification from a Finite Set}, 
      author={Nicolas Chatzikiriakos and Andrea Iannelli},
      journal={IEEE Control Systems Letters}, 
      year={2024},
      volume={8},
      number={},
      pages={2751-2756},}

@article{mania2022active,
  title={Active learning for nonlinear system identification with guarantees},
  author={Mania, Horia and Jordan, Michael I and Recht, Benjamin},
  journal={J. of Machine Learning Research},
  volume={23},
  number={32},
  pages={1--30},
  year={2022}
}

@InProceedings{foster20a,
  title = 	 {Learning nonlinear dynamical systems from a single trajectory},
  author =       {Foster, Dylan and Sarkar, Tuhin and Rakhlin, Alexander},
  booktitle = 	 {Proceedings of the 2nd Conf. on Learning for Dynamics and Control},
  pages = 	 {851--861},
  year = 	 {2020},
  volume = 	 {120},
  publisher =    {PMLR},
}

@inproceedings{chatzikiriakos2025convex,
  title   = {Hidden Convexity in Active Learning: A Convexified Online Input Design for {ARX} Systems},
  author={Chatzikiriakos, Nicolas and Song, Bowen and Rank, Philipp and Iannelli, Andrea},
  booktitle={IEEE 64th Conf. on Decision and Control (CDC)},
  year={2025}
}

@inbook{vershynin2012, 
    place={Cambridge}, 
    title={Introduction to the non-asymptotic analysis of random matrices}, 
    booktitle={Compressed Sensing: Theory and Applications}, 
    publisher={Cambridge University Press}, 
    author={Vershynin, Roman},
    editor={Eldar, Yonina C. and Kutyniok, GittaEditors}, 
    year={2012}, 
    pages={210–268}
}

@article{waarde2023,
    author = {van Waarde, Henk J. and Camlibel, M. Kanat and Eising, Jaap and Trentelman, Harry L.},
    title = {Quadratic Matrix Inequalities with Applications to Data-Based Control},
    journal = {SIAM J. on Control and Optimization},
    volume = {61},
    number = {4},
    pages = {2251-2281},
    year = {2023},
    doi = {10.1137/22M1486807},
}

@article{umenberger2019,
  title={Robust exploration in linear quadratic reinforcement learning},
  author={Umenberger, Jack and Ferizbegovic, Mina and Sch{\"o}n, Thomas B and Hjalmarsson, H{\aa}kan},
  journal={Advances in Neural Information Processing Systems},
  volume={32},
  year={2019}
}

@ARTICLE{waarde2022,
  author={van Waarde, Henk J. and Camlibel, M. Kanat and Mesbahi, Mehran},
  journal={IEEE Transactions on Automatic Control}, 
  title={From Noisy Data to Feedback Controllers: Nonconservative Design via a Matrix S-Lemma}, 
  year={2022},
  volume={67},
  number={1},
  pages={162-175},
  doi={10.1109/TAC.2020.3047577}
}

@INPROCEEDINGS{berberich2020,
  author={Berberich, Julian and Koch, Anne and Scherer, Carsten W. and Allgöwer, Frank},
  booktitle={2020 American Control Conf. (ACC)}, 
  title={Robust data-driven state-feedback design}, 
  year={2020},
  volume={},
  number={},
  pages={1532-1538},
  doi={10.23919/ACC45564.2020.9147320}
}

@ARTICLE{dePersis2020,
  author={De Persis, Claudio and Tesi, Pietro},
  journal={IEEE Transactions on Automatic Control}, 
  title={Formulas for Data-Driven Control: Stabilization, Optimality, and Robustness}, 
  year={2020},
  volume={65},
  number={3},
  pages={909-924},
  doi={10.1109/TAC.2019.2959924}
}

@ARTICLE{braendle2025a,
  author={Brändle, Felix and Allgöwer, Frank},
  journal={IEEE Control Systems Letters}, 
  title={A System Parameterization for Direct Data-Driven Estimator Synthesis}, 
  year={2025},
  volume={9},
  number={},
  pages={1225-1230},
  doi={10.1109/LCSYS.2025.3580035}
}

@article{martin2023a,
title = {Gaussian inference for data-driven state-feedback design of nonlinear systems},
journal = {IFAC-PapersOnLine},
volume = {56},
number = {2},
pages = {4796-4803},
year = {2023},
issn = {2405-8963},
doi = {https://doi.org/10.1016/j.ifacol.2023.10.1245},
author = {Tim Martin and Thomas B. Schön and Frank Allgöwer},
}

@INPROCEEDINGS{braendle2025b,
  author={Brändle, Felix and Allgöwer, Frank},
  booktitle={2025 IEEE 64th Conf. on Decision and Control (CDC)}, 
  title={Data-driven Estimator Synthesis with Instantaneous Noise}, 
  year={2025},
  volume={},
  number={},
  pages={367-372},
  doi={10.1109/CDC57313.2025.11312215}
}

@Article{Weiland1994,
  author  = {Weiland, Siep and Scherer, Carsten},
  journal = {Lecture Notes, Dutch Institute for Systems and Control, Delft},
  title   = {Linear {Matrix} {Inequality} in {Control}},
  year    = {2000},
}

@inproceedings{Yingying2024,
  title={Learning the uncertainty sets of linear control systems via set membership: A non-asymptotic analysis},
  author={Li, Yingying and Yu, Jing and Conger, Lauren and Kargin, Taylan and Wierman, Adam},
  booktitle={41st International Conf. on Machine Learning},
  year={2024}
}

@article{chatzikiriakos2026a,
title = {End-to-end guarantees for indirect data-driven control of bilinear systems with finite stochastic data},
journal = {Automatica},
volume = {187},
pages = {112908},
year = {2026},
issn = {0005-1098},
doi = {https://doi.org/10.1016/j.automatica.2026.112908},
author = {Nicolas Chatzikiriakos and Robin Str{\"a}sser and Frank Allg{\"o}wer and Andrea Iannelli},

}

@article{MILANESE2004957,
title = {Set Membership identification of nonlinear systems},
journal = {Automatica},
volume = {40},
number = {6},
pages = {957-975},
year = {2004},
author = {Mario Milanese and Carlo Novara},
}

@article{mosek,
   author = {{MOSEK ApS}},
   title = {The MOSEK optimization toolbox for MATLAB manual. Version 9.0.},
   year = {2019},
 }

@article{care2017finite,
  title={Finite-sample system identification: An overview and a new correlation method},
  author={Care, Algo and Cs{\'a}ji, Bal{\'a}zs Cs and Campi, Marco C and Weyer, Erik},
  journal={IEEE Control Systems Letters},
  volume={2},
  number={1},
  pages={61--66},
  year={2017},
  publisher={IEEE}
}

\end{document}